\documentclass[12pt]{article}
\usepackage{dblfloatfix}
\usepackage{marginnote}
\usepackage{listings}
\usepackage{algorithm}
\usepackage{algpseudocode}
\usepackage[english]{babel}

\usepackage[letterpaper,top=2cm,bottom=2cm,left=3cm,right=3cm,marginparwidth=1.75cm]{geometry}

\usepackage[utf8]{inputenc}
\usepackage{mathrsfs}
\usepackage{dsfont}
\usepackage{amsfonts}
\usepackage{complexity}
\usepackage{latexsym}
\usepackage{amssymb}    
\usepackage{amsmath}
\usepackage{graphicx}
\usepackage[colorlinks=true, allcolors=blue]{hyperref}

	\newcommand{\bbF}{\mathbb{F}}

	\newcommand{\bbN}{\mathbb{N}}
	
\newcommand{\bbQ}{\mathbb{Q}}

\newcommand{\mcA}{\mathcal{A}}	\newcommand{\mcB}{\mathcal{B}}
\newcommand{\mcC}{\mathcal{C}}

\newcommand{\mcM}{\mathcal{M}}	
\newcommand{\mcO}{\mathcal{O}}

\newcommand{\mfI}{\mathfrak{I}}

	\newcommand{\wtL}{\widetilde{L}}

\newcommand{\eps}{\epsilon}

\newcommand{\ceil}[1]{\lceil #1 \rceil}

\newcommand{\F}{\mathbb{F}}

\newcommand{\rank}{\operatorname{rank}}

\newcommand{\Det}{\operatorname{\mathsf{Det}}}

\renewcommand{\P}{\mathsf{P}}

\renewcommand{\epsilon}{\varepsilon}

\renewcommand{\epsilon}{\varepsilon}
\newcommand{\ignore}[1]{}
\usepackage{cleveref}
\usepackage{babel}
\usepackage{amsthm}
\usepackage{thmtools,thm-restate}

\newtheorem{definition}{Definition}[section]
\newtheorem{theorem}{Theorem}[section]
\newtheorem{corollary}{Corollary}[theorem]
\newtheorem{lemma}[theorem]{Lemma}

\newtheorem{fact}[theorem]{Fact}
\newtheorem{claim}[theorem]{Claim}
\newtheorem{remark}[theorem]{Remark}

\title{Maximum Matching and Related Problems in Catalytic Logspace}
\author{Srijan Chakraborty, Samir Datta, Aryan Kusre, Partha Mukhopadhyay, \\Amit Sinhababu\\
Chennai Mathematical Institute, India\\
\texttt\{srijanc, sdatta, aryank, partham, amitks\}@cmi.ac.in}

\begin{document}
\maketitle

\abstract 
Understanding the power of space-bounded computation with access to catalytic space has been an important theme 
in complexity theory over the recent years. One of the key algorithmic results in this area is that bipartite maximum matching can be computed in catalytic logspace ($\CL$) with a polynomial-time bound ($\CL\P$) \cite{AM}. 

In this paper, we show that we can construct a \emph{maximum matching} in \emph{general graphs} in $\CL$, and, in fact, in $\CLP$. We first show that the size of a \emph{maximum matching} in \emph{general graphs} can be determined in $\CL$. Our algorithm is based on the linear-algebraic algorithm for maximum matching by Geelen \cite{Geelen2000}. We then show that this algorithm, along with some new ideas, can be used to \emph{find} a maximum matching in general graphs. Using a similar algorithm of Geelen \cite{Geelen1999}, we also solve the \emph{maximum rank completion problem} in $\CL\P$, which was previously known to be solvable in deterministic polynomial time \cite{Geelen1999}. This problem turns out to be equivalent to the \emph{linear matroid intersection} problem \cite{Murota1995Mixed} which has been shown to be in $\CLP$ by \cite{AAV26}. Finally, using a PTAS algorithm \cite{BlaserJindalPandey2018} for approximating the rank in Edmond's problem, we derive a $\CLP$ algorithm that can approximate the rank given by any instance of the \emph{Edmond's problem} upto a factor of $(1-\eps)$ for any $\eps\in(0,1)$. An application of this is a $\CLP$ bound for approximating the maximum independent matching size in the \emph{linear matroid matching} problem.


\section{Introduction}\label{sec:intro}
The study of space-bounded computation is a central theme in complexity theory, with a long line of work investigating the power and limitations of logarithmic-space algorithms. A recent direction in this area focuses on models that augment classical space bounds with auxiliary resources, such as catalytic space, where a large workspace is available but must be returned to its initial state at the end of the computation. This model, introduced by Buhrman, Cleve,  Kouck{\'y}, Loff, and Speelman\cite{BuhrmanCL2014}, has led to surprising algorithmic developments and a refined understanding of the role of reversibility and space reuse in computation.

A major goal in this line of research is to identify natural problems that can be solved efficiently within catalytic logspace ($\CL$), particularly with polynomial-time bounds ($\CL\P$). While several foundational results have established the basic properties of this model, progress on natural combinatorial problems has been more recent. Notably, the breakthrough result of \cite{AM} shows that bipartite maximum matching can be computed in catalytic logspace with polynomial time, providing one of the first nontrivial examples of a classical graph problem lying in $\CL\P$. Over the recent years, there have been many results surrounding catalytic computation such as \cite{BuhrmanCL2014,Pyne,CookLiMertzPyne2025,CGMPS,AAV26,ACD,Edenhofer,CDKMR}. A flowchart for various complexity classes would be: $\Log \subseteq \NL\subseteq\TC^1\subseteq\NC\subseteq\P\subseteq\ZPP$ where the only known bounds for catalytic logspace are $\TC^1\subseteq\CL\subseteq\ZPP$. 

Matching problems occupy a central position in combinatorial optimization and theoretical computer science. The maximum matching problem, in both bipartite and general graphs, has been extensively studied, with classical polynomial-time algorithms based on augmenting paths, combinatorial structures such as blossoms, and algebraic formulations \cite{Tutte1947,Edmonds1965,HopcroftKarp1973, Lovasz1979}. In particular, algebraic approaches based on the Tutte matrix and rank computations, such as those developed by Jim Geelen, provide a powerful framework for reasoning about matchings and their generalizations \cite{Geelen1999, Geelen2000}.

In this paper, we extend the reach of catalytic logspace algorithms for matching problems. We show that a maximum matching in general graphs can be computed in $\CL\P$, thereby generalizing the bipartite case. To show that bipartite matching is in $\CLP$, \cite{AM} uses isolating weights which they sample from the catalytic tape. They can solve the problem in one shot if a given weight assignment is isolating, and otherwise they find a threshold edge by finding appropriate augmenting paths and can free up space. This technique breaks down in general graphs, since  finding augmenting paths becomes much more difficult, as one might come across blossoms while doing so. To avoid these issues, we build on the algebraic techniques of Geelen \cite{Geelen2000} and adapt them to the catalytic setting, carefully managing space reuse. To search for a maximum matching, we again use isolating weights, but not in the conventional sense, in particular, we do not isolate a min-weight matching. As a further application, we show that the \emph{maximum rank matrix completion} problem can also be solved in $\CL\P$ \cite{Geelen1999}. This problem, previously known to admit deterministic polynomial-time algorithms, captures a common generalization of linear matroid intersection and bipartite matching \cite{AAV26,AM}, and thus serves as a unifying framework for these problems. 

Both maximum matching, and the matrix completion problems are special cases of \emph{symbolic determinant identity testing} (SDIT), which is also known as \emph{Edmonds' problem} \cite{Edmonds1965,Lovasz1979}. Given a matrix 
$A=\sum_{i\in[m]} A_i x_i$, the SDIT problem asks for the rank of $A$. There is a simple randomized polynomial-time algorithm for this problem: one can obtain a matrix of maximum rank by substituting random values for $x_1, x_2, \ldots, x_m$ from a sufficiently large set from the field $\F$, using the Polynomial Identity Testing lemma \cite{Zippel1979,Schwartz1980,DeMilloLipton1978}.

In \cite{BlaserJindalPandey2018}, Bl\"{a}ser, Jindal, and Pandey study this problem of computing the rank of the symbolic matrix $A$, and are able to get a PTAS for the same. However, exact computation of the commutative rank of a given matrix space in deterministic polynomial time is the main problem of symbolic identity testing (SDIT) which is directly related to lower bounds \cite{KabanetsImpagliazzo2004}. 

Extending the ideas in \cite{BlaserJindalPandey2018}, we derive a $\CLP$ algorithm that can approximate the maximum rank of a matrix in any matrix space upto a factor of $(1-\eps)$ for any $\eps\in(0,1)$. An interesting application of this is a $\CLP$ algorithm for approximating the size of a maximum independent matching in the linear matroid matching problem.

Our results provide new evidence for the power of catalytic space in enabling efficient algorithms for algebraic and combinatorial problems. More broadly, they suggest that linear-algebraic techniques may play a fundamental role in understanding the capabilities of space-bounded computation with reversible auxiliary memory. This also gives a $\CL$ bound for a couple of problems, that we do not yet know how to solve in $\NC$, namely maximum matching size in general graphs, matrix rank completion, $(1-\eps)$ approximating Edmond's rank. Moreover $\CLP$ also turns out to be the first well studied subclass of $\P$ where we can find a maximum matching in general graphs.
\subsection*{Our results}

The maximum matching problem in general graphs admits a well-known algebraic formulation via the Tutte matrix \cite{Tutte1947}. Given a graph $G$, the Tutte matrix is a symbolic skew-symmetric matrix in which each edge $(i,j)$ is associated with an indeterminate $x_{i,j}$ (with a sign), and non-edges correspond to zero entries.

The fundamental connection shows that the rank of this matrix (over a suitable field) is twice the size of a maximum matching in $G$. This characterization enables the use of linear-algebraic techniques for computing matchings, including randomized algorithms based on substituting indeterminates with field elements. In particular, combining this approach with fast parallel algorithms for matrix rank yields an $\RNC$ upper bound for maximum matching in general graphs \cite{MulmuleyVaziraniVazirani1987}, placing the problem among the central examples of efficiently parallelizable combinatorial problems. From the perspective of derandomization and pinning down the parallel complexity, putting matching in NC has been an outstanding question.

The last decade saw advances in understanding the parallel complexity of matching. In particular, building on the breakthrough quasi-$\NC$ algorithm for bipartite matching due to Fenner, Gurjar, and Thierauf \cite{FennerGurjarThierauf2016}, follow-up work by Svensson and Tarnawski \cite{SvenssonTarnawskiVassilevskaWilliams2017} obtained quasi-$\NC$ algorithms for maximum matching in general graphs.

Despite these advances, a $\CL$ algorithm to find a perfect matching in general graphs has remained elusive. We settle this question by proving the following theorem. 

\begin{theorem}\label{thm:main1}
 Maximum matchings in a general graph $G$ can be constructed in $\CLP$. 
\end{theorem}

The maximum rank matrix completion problem asks, given a matrix with some entries fixed and others unspecified, to assign values to the unspecified entries so as to maximize the rank of the resulting matrix. This problem admits a natural algebraic interpretation and captures several combinatorial optimization problems as special cases. In particular, Geelen showed that the problem can be solved in deterministic polynomial time via a greedy algorithm that incrementally assigns values while preserving rank growth \cite{Geelen1999}. The problem generalizes classical questions such as bipartite matching and linear matroid intersection, and plays a central role in algebraic approaches to matching and related problems. Extending the techniques used to prove Theorem \ref{thm:main1}, we show the following result. 

\begin{theorem}\label{thm:main2}
The maximum rank completion problem can be solved in $\CLP$. 
\end{theorem}

In \cite{BlaserJindalPandey2018}, Bl\"{a}ser, Jindal, and Pandey study the problem of computing the commutative rank of a matrix space, a fundamental quantity that captures the maximum rank achievable by linear combinations of given matrices $A_1, A_2, \ldots, A_n$. More precisely, the rank of a matrix space $\mathcal{A}=\langle A_1, A_2, \ldots, A_n\rangle$ over a sufficiently large field $\mathbb{F}$ is the maximum rank of a matrix in $\mathcal{A}$. If we consider the formal linear combination $A=\sum_{i=1}^n A_i x_i$, then we can compute a matrix of maximum rank by substituting random values to $x_1, x_2, \ldots, x_n$ from $\mathbb{F}$ by applying the following Polynomial Identity Testing Lemma \cite{Zippel1979, Schwartz1980, DeMilloLipton1978}:
\begin{lemma}\label{SZ}
 Let
$
P \in R\left[x_1, x_2, \ldots, x_n\right]
$
be a non-zero polynomial of total degree $d \geq 0$ over an integral domain $R$. Let $S$ be a finite subset of $R$ and let $r_1, r_2, \ldots, r_n$ be selected at random independently and uniformly from $S$. Then
$
\operatorname{Pr}\left[P\left(r_1, r_2, \ldots, r_n\right)=0\right] \leq \frac{d}{|S|}
$
\end{lemma}

As already mentioned, exact computation of the commutative rank of a given matrix space in deterministic polynomial time is hard \cite{KabanetsImpagliazzo2004}.
Nevertheless, Bl\"{a}ser, Jindal, and Pandey \cite{BlaserJindalPandey2018} could show a deterministic \emph{polynomial-time approximation scheme} (PTAS) that computes a $(1-\eps)$-approximation of the rank for any fixed $\eps>0$. Their approach combines algebraic and combinatorial techniques to bypass randomness. Our main result in this context is the following.

\begin{theorem}\label{thm:main3}
Given a matrix space $\mathcal{A}=\langle A_1, A_2, \ldots, A_n\rangle$ by a set of matrices $A_1, A_2, \ldots, A_n$ and $\epsilon >0$, we can compute a matrix $A\in \mathcal{A}$ whose rank is at least $1-\epsilon$ times the rank of $\mathcal{A}$, in $\CLP$. 
\end{theorem}

Since our main result uses ideas from Geelen's algebraic matching algorithm \cite{Geelen2000}, we provide a brief overview of his result. Let $T$ be the Tutte matrix of $G=(V,E)$ where $|V|=n$. For any evaluation 
of $T'$, let $\rank(T')$ denote its rank.  
Let $D(T')$ denote the corresponding set of deficiency, i.e., the set of vertices $x$ such that $\rank(T'\setminus \{x\})=\rank(T')$.  Intuitively, the set of vertices corresponds to the Gallai--Edmonds decomposition induced by $T'$ in a linear algebraic sense. Then, Geelen \cite{Geelen2000} defines a preorder $T'_1 \preceq T'_2$ either
\[
\rank(T'_1)< \rank(T'_2)
\quad \text{or} \quad
(\rank(T'_2)=\rank(T'_1)
\quad \text{and}\quad
D(T'_1)\subseteq D(T'_2)).
\]
If both $\rank(T'_1)=\rank(T'_2)$ and $D(T'_1)=D(T'_2)$ then $T'_1=T'_2$. 
The algorithm incrementally assigns values to the indeterminates corresponding to the edges and constructs a sequence $T'_1\prec T'_2\prec \cdots\prec T'_m$ where the final matrix achieves the rank of the Tutte matrix. 

More precisely, to go from $T'_i$ to $T'_{i+1}$, the algorithm selects an edge $(i,j)$ and replaces its value by $a\in\{1,2, \ldots, n\}$ such that $T'_i\prec T'_{i+1}$. If such an assignment does not exist then the algorithm already has found the evaluated matrix achieving the true rank. After finding the size of the maximum matching, they use self-reducibility of matching to search for one. But under our space bounded constraints, the search for a maximum matching in the same way as above, does not work out. Notice that \cite{ACD} shows a $\CL$ reduction from search to weighted decision, that works in the matching context as well, but it is also not clear how to solve the weighted decision using our ideas. To construct a maximum matching of the graph, we need some substantially new ideas. We again sample weight assignments from the catalytic tape, as \cite{AM}, but we are not able to find min-isolating weights. In stead, we have a non-zero polynomial, the least degree of which turns out to be the weight of the matching that we attempt to isolate. Smaller weighted matchings may disappear in this polynomial. Finally, if we are not able to find a matching even in this way, we show that we are able to save enough space to construct a matching in polynomial space and time, and later restore the catalytic tape. We heavily use the \emph{`compress or compute'} paradigm.

\section{Preliminaries}
 
\paragraph*{Matchings}

Let $G = (V,E)$ be a simple undirected graph. A \emph{matching} $M \subseteq E$ is a set of pairwise vertex-disjoint edges. A matching is said to be \emph{maximum} if it has the largest possible cardinality among all matchings in $G$.

We denote by $\nu(G)$ the size of a maximum matching in $G$, i.e.,
\[
\nu(G) = \max \{ |M| : M \subseteq E \text{ is a matching} \}.
\]

\paragraph*{Gallai--Edmonds Decomposition}

Let $G = (V,E)$ be a simple graph. Define:
\begin{enumerate}
	\item $D(G)$: the set of vertices that are unmatched in at least one maximum matching of $G$,
	\item $A(G)$: the set of vertices in $V \setminus D(G)$ that have a neighbor in $D(G)$,
	\item $C(G) := V \setminus (A(G) \cup D(G))$.
\end{enumerate}

The partition $(D(G), A(G), C(G))$ is called the \emph{Gallai--Edmonds decomposition} of $G$.

This decomposition satisfies the following properties:
\begin{enumerate}
	\item Each connected component of $G[D(G)]$ is factor-critical.
	\item $G[C(G)]$ has a perfect matching.
	\item In every maximum matching:
	\begin{enumerate}
		\item all vertices of $A(G)$ are matched to vertices in $D(G)$,
		\item each component of $G[D(G)]$ has exactly one unmatched vertex.
	\end{enumerate}
\end{enumerate}

\paragraph*{Linear Matroids}

Let $\mathbb{F}$ be a field. A \emph{linear matroid} is a matroid $M = (E, \mathcal{I})$ where $E$ is a finite set and there exists a matrix $A \in \mathbb{F}^{r \times |E|}$ such that a subset $I \subseteq E$ is independent (i.e., $I \in \mathcal{I}$) if and only if the corresponding set of columns of $A$ is linearly independent over $\mathbb{F}$.

\paragraph*{Linear Matroid Intersection}

Let $M_1 = (E, \mathcal{I}_1)$ and $M_2 = (E, \mathcal{I}_2)$ be two linear matroids represented over the same ground set $E$. The \emph{linear matroid intersection problem} asks to find a set
\[
I \subseteq E \quad \text{such that } I \in \mathcal{I}_1 \cap \mathcal{I}_2
\]
of maximum cardinality.

\paragraph*{Linear Matroid Parity}

Let $M = (E, \mathcal{I})$ be a linear matroid, and suppose the ground set $E$ is partitioned into disjoint pairs:
\[
E = \{ e_1, f_1, e_2, f_2, \dots, e_m, f_m \}.
\]

The \emph{linear matroid parity problem} asks to find a maximum-size collection of pairs such that the union of the selected elements is independent in $M$. Formally, find an index set $S \subseteq [m]$ maximizing $|S|$ such that
\[
\{ e_i, f_i : i \in S \} \in \mathcal{I}.
\]

\paragraph*{Linear Matroid Matching}

Let $M = (E, \mathcal{I})$ be a linear matroid, and let $G = (E, F)$ be a graph whose vertex set is the ground set of the matroid.

A set of edges $S \subseteq F$ is a \emph{matroid matching} if:
\begin{enumerate}
	\item $S$ is a matching in the graph $G$, and
	\item the set of vertices covered by $S$ is independent in the matroid $M$.
\end{enumerate}

The \emph{linear matroid matching problem} asks to find a matroid matching of maximum cardinality.

\paragraph*{Determinant}
For a square matrix $A \in \mathbb{F}^{n \times n}$, the determinant is defined as:
\[ \Det(A) = \sum_{\sigma \in S_n} \text{sgn}(\sigma) \prod_{i=1}^{n} A_{i, \sigma(i)} \]
where $S_n$ is the symmetric group and $\text{sgn}(\sigma)$ is the parity of the permutation.

\paragraph*{Pfaffian}
For a skew-symmetric matrix $A \in \mathbb{F}^{2n \times 2n}$ where $A^T = -A$, the Pfaffian is defined as:
\[ \Pf(A) = \frac{1}{2^n n!} \sum_{\sigma \in S_{2n}} \text{sgn}(\sigma) \prod_{i=1}^{n} A_{\sigma(2i-1), \sigma(2i)} \]
It holds that $\Det(A) = \Pf(A)^2$.

\paragraph*{Catalytic Computation}

\begin{definition}
	A catalytic Turing machine $M$ with space $s(n)$ and catalytic space $c(n)$ is a Turing machine that has a read-only input tape of length $n$, write-only output tape, and $s(n)$ space bounded read-write work tape, and a catalytic tape of size $c(n)$. 
	We say that $M$ computes a function $f$ if for every $x \in\{0,1\}^n$ and $\tau \in\{0,1\}^{c(n)}$, the result of executing $M$ on input $x$ with initial catalytic tape $\tau$ i.e. $M(x,\tau)$ satisfies:
	\begin{enumerate}
		\item $M$ halts with $f(x)$ on the output tape.
		\item $M$ halts with the catalytic tape consisting of $\tau$.
	\end{enumerate} 
	$\CSPACE[s(n),c(n)]$ is the family of functions computable by such a Turing machine. Similarly, $\CTISP[t(n),s(n),c(n)]$ is the family of functions computable by such a Turing machine with the further restriction that the machine simultaneously runs in time $t(n)$.
\end{definition}
\begin{definition}[Catalytic Logspace]
	We define catalytic logspace as 
	\[\CL=\cup_{k\in\bbN}\CSPACE[k\log n,n^k]\]
	Moreover, $\CL\P$ is defined as the set of functions computable by a $\CL$ machine that simultaneously runs in polynomial time.
\end{definition}
In particular, we know that $\CL\P=\CL\cap\P$ \cite{CookLiMertzPyne2025}.
To compute the rank of various matrices, we shall use the fact that $\TC^1\subseteq\CL$ \cite{BuhrmanCL2014}. It is well known that Determinants and Pfaffians are computable in $\TC^1$ and thus in $\CL$.

\section{Maximum Matching in CLP}\label{sec:matching}
In this section, we prove the decision version of Theorem \ref{thm:main1}. 

\subsection{A linear algebraic formulation}
\begin{fact}
Let $A$ be any skew symmetric matrix. Then
\begin{enumerate}
	\item $\rank A$ is even.
	\item The rank of $A$ is equal to the maximum size of a non-singular principal minor/submatrix of $A$.
\end{enumerate}
\end{fact}

Let $G=(V,E)$ be a simple graph, $|V|=n,|E|=m$, $x_1,\ldots x_m$ be edge variables. We now state the main result in \cite{Geelen2000} tailored made to our application. The Tutte matrix $T$ of $G$ is the $n\times n$ matrix given by:
\[
T_{uv} = \begin{cases}
	\; x_{i} & u<v,\ i=\{u,v\}\in E,\\[2mm]
	\; -x_{i} & u>v,\ i=\{u,v\}\in E,\\[2mm]
	\; 0 & \text{otherwise.}
\end{cases}
\] 
\begin{theorem}[Geelen, \cite{Geelen2000}]\label{geelen}
	Let $T$ be the Tutte matrix of a simple graph $G=(V,E)$, and let $T'$ be an evaluation of $T$. Then either $\rank T'=\rank T$ and $D(T')=D(T)=D(G)$, or one of the following two statements hold:
	\begin{itemize}
		\item $\exists i\in E, a\in [n^{10}]$ such that $\rank T'_{x_i\gets a}=\rank T'+2$.
		\item $\exists i\in E, a\in [n^{10}]$ such that $\rank T'_{x_i\gets a}=\rank T'\text{and}~ D(T')\subsetneq D(T'_{x_i\gets a})$.
	\end{itemize}
\end{theorem}
Recall that $T'_{x_i\gets a}$ is the matrix obtained from $T'$ by changing the value corresponding to the edge variable $x_i$ by $a$. Moreover, as already mentioned in Section \ref{sec:intro} that $D(T)$ is the set of vertices which are deficient in terms of rank. Similarly $D(G)$ is the set of vertices in $G$ which are missed by some maximum matching.  
\begin{remark}\label{rmk:remark1}
	In \cite{BuhrmanCL2014} it is shown that $\TC^1\subseteq \CL$. This already shows the determinant and matrix rank computation can be performed in $\CL$. Clearly, the membership question for the set $D$ logspace reduces to the rank computation which is in $\CL$. 
\end{remark}

Given an evaluation $T'$ of the Tutte matrix $T$, the following algorithm decides whether the rank of $T'$ is the highest possible rank (or the rank of $T$). 

\begin{algorithm}
	\caption{IsMaxRank$(T',T,G=(V,E))$}\label{alg1}
	\begin{algorithmic}[1]
		\State $k\gets \rank T'$ \Comment{Can be done in $\CL$}
		\State $maxrank\gets true$
		\For {$i\in E,a\in [n^{10}]$}
		\State $T''=T'_{x_i\gets a}$
		\If {$(k<\rank T'')\vee (k=\rank T''\wedge D(T')\subsetneq  D(T''))$}\Comment{Can be checked in $\CL$}
		\State $maxrank\gets false$
		\EndIf
		\EndFor
		\State \Return $maxrank$
	\end{algorithmic}
\end{algorithm}
\begin{lemma}
	Suppose, we are given an evaluation of $T$, namely $T'$. Then we can check in $\CL$ if $\rank T=\rank T'$ or not.
\end{lemma}
\begin{proof}
	\Cref{alg1} gives the required procedure, the proof of which directly follows from \Cref{geelen} and \Cref{rmk:remark1}.
\end{proof}

Next we show two uniqueness lemmas, which will turn out to be essential ingredients in the design of our catalytic algorithms.

\begin{lemma}[$2A$]\label{2A}
	Let $T'$ be an evaluation of $T$ and $i\in E, a\in [n^{10}]$ such that $\rank T'_{x_i\gets a}=\rank T'+2$, and in $T'$, $x_i$ is $a_i$. Then $\forall \tilde a\ne a_i$, we have $\rank T'_{x_i\gets \tilde a}=\rank T'+2$.
\end{lemma}
\begin{proof}
	Let $k=\rank T'$. 
		Let $X\subseteq V, |X|=k+2$ be a maximal principal minor of $T$ that is non-singular in $T'_{x_i\gets a}[X]$. Notice that $L(y)=\Pf (T'_{x_i\gets y}[X])$ is linear in $y$. Since $L(a)\ne 0, L(a_i)=0$, clearly $\forall \tilde a\ne a_i, L(\tilde a)\ne 0$ i.e $T'_{x_i\gets \tilde a}[X]$ is non-singular.
	 Furthermore, on changing two entries of $T'$ (i.e. by changing the variable $x_i$) the rank can go up by atmost $2$. Therefore, $\rank T'_{x_i\gets \tilde a}=\rank T'+2$ $\forall \tilde a\ne a$.
\end{proof}

\begin{lemma}[$2B$]\label{2B}
	Let $T'$ be an evaluation of $T$ and $i\in E, a\in [n^{10}]$ such that $\rank T'_{x_i\gets a}=\rank T'=k$ where in $T'$, $x_i$ is set to $a_i$. Suppose $\exists u\in D(T'_{x_i\gets a})\setminus D(T')$. Then $\forall \tilde a\ne a_i$, we have $u\in D(T'_{x_i\gets \tilde a})$. 
\end{lemma}
\begin{proof}
	From the given assumptions, we have a $X_u\subseteq V,|X_u|=k$ such that $T'_{x_i\gets a}[X_u]$ is non-singular and $u\notin X_u$. Since $u\notin D(T')$, $T'[X_u]$ is singular.
	\begin{claim}
		$\forall \tilde a\ne a_i$, $\rank T'_{x_i\gets \tilde a}\le k$.
	\end{claim}
	\begin{proof}
		Suppose $\rank T'_{x_i\gets \tilde a}>k$ for some $\tilde a$. Therefore, $\tilde a\notin\{ a,a_i\}$. Let $X$ be the principal minor of size larger than $k$ such that $T'_{x_i\gets \tilde a}[X]$ is non-singular. But then $\wtL (y)=\Pf(T'_{x_i\gets y}[X])$ is linear and $\wtL(a)=\wtL(a_i)=0$. Therefore, $\wtL(\tilde a)=0$ gives a contradiction.
	\end{proof}
	 Again $L(y)=\Pf (T'_{x_i\gets y}[X_u])$ is linear in $y$, such that $L(a)\ne 0, L(a_i)=0$.
	 Therefore, $\forall \tilde a\ne a_i$ we have $L(\tilde a)\ne 0$ i.e $T'_{x_i\gets \tilde a}[X_u]$ is non-singular and by the above claim $X_u$ is of maximum possible size (i.e. $\rank T'_{x_i\gets \tilde a}=|X_u|$), thus $u\in D(T'_{x_i\gets \tilde a})$.
\end{proof}

\subsection{A compress or compute algorithm}



Now we show that the size of a maximum matching in a general graph can be computed in $\CL$. Let $G$ be a graph on $n$ vertices and $m$ edges, and let $T$ be its Tutte matrix with variables $x_1,\ldots,x_m$ corresponding to the edges.

We sample an assignment $a_1,\ldots,a_m$ for these variables from the catalytic tape. Either this assignment maximizes the rank of the evaluated matrix $T'$, or it does not. We can verify this using \Cref{alg1}. If the rank is indeed maximized, we proceed to restore the modified catalytic tape.

Otherwise, if the rank is not maximized, we consider two cases:
\begin{itemize}
    \item[(a)] Suppose that by perturbing one of the variables $x_i$, the rank of $T'$ can be increased. In this case, by \Cref{2A}, for all values other than $a_i$, the rank increases. Therefore, we discard $a_i$ from the catalytic tape and instead store the index $i$ along with the rank. Later, we can recover $a_i$ as the unique value for which the rank remains unchanged.
    
    \item[(b)] Otherwise, perturbing some variable $x_i$ does not increase the rank, but enlarges the set $D(T')$, i.e., there exists a vertex $u$ that gets added to $D(T')$. By \Cref{2B}, $a_i$ is the unique value for which $u \notin D(T')$. Thus, we again discard $a_i$ and store its index $i$ together with $u$. As before, $a_i$ can be recovered when needed using this characterization.
\end{itemize}

Finally, if the sampled assignment already yields the maximum rank, we are done. Otherwise, we have freed enough space to recompute the rank from scratch using the algorithm of \cite{Geelen2000}. A formal proof follows.

\begin{algorithm}
	\caption{ProcessAssignment$(\mcC,T,G=(V,E))$}\label{alg2}
	\begin{algorithmic}[1]
		\State Input: We are given $\mcC=(a_1,\ldots,a_m)$ as an initial catalytic string, where $a_i\in [n^{10}]\forall i\in E$. 
		\State $T'\gets T_{\forall i\in E: x_i\gets a_i}$, $k\gets \rank T'$
	    \If {IsMaxRank$(T',T,G)$}\Comment{Call \Cref{alg1}}
		\State \Return $\nu(G)=\dfrac k 2, D(G)=D(T')$
		\EndIf
		\State $case\gets 2$
		\For {$i\in E, a\in [n^{10}]$}
		\State $T''\gets T'_{x_i\gets a}$
		\If {$\rank T''=k+2$}
		\State $case\gets 2A$
		\State Reset $\mcC\gets (2A,i,k,a_1,\ldots,a_{i-1},a_{i+1},\ldots,a_m)$.
		\State \Return \Comment{Terminate the algorithm}
		\EndIf
		\EndFor
		\For {$i\in E, a\in [n^{10}]$}
		\State $T''\gets T'_{x_i\gets a}$
		\If {$(\rank T''=k)\wedge (D(T')\subsetneq  D(T''))$}
		\State $case\gets 2B$
		\State Let $u\in D(T'')\setminus D(T')$ \Comment{Pick any such $u$ arbitrarily}
		\State Reset $\mcC\gets  (2B,i,u,a_1,\ldots,a_{i-1},a_{i+1},\ldots,a_m)$.
		\State \Return \Comment{Terminate the algorithm}
		\EndIf
		\EndFor
	\end{algorithmic}
\end{algorithm}
\begin{algorithm}
	\caption{RestoreAssignment$(\mcC,T,G=(V,E))$}\label{alg3}
	\begin{algorithmic}[1]
		\State Input: We are given $\mcC=(case,j,\beta,a_1,\ldots,a_{j-1},a_{j+1},\ldots,a_m), a_i\in [n^{10}]\forall i\in E$. 
		\State $T'\gets T_{\forall i\in E\setminus \{j\}: x_i\gets a_i}$
		\If {$case=2A$}
		\State $k\gets \beta$
		\For {$a\in [n^{10}]$}
		\If {$k=\rank T'_{x_j\gets a}$}
		\State Reset $\mcC\gets (a_1,\ldots,a_{j-1},a_j\gets a,a_{j+1},\ldots,a_m)$ 
		\EndIf
		\EndFor
		\Else
		\State $u\gets \beta$
		\For {$ a\in [n^{10}]$}
		\If {$u\notin D(T'_{x_j\gets a})$}
		\State Reset $\mcC\gets (a_1,\ldots,a_{j-1},a_j\gets a,a_{j+1},\ldots,a_m)$ 
		\EndIf
		\EndFor
		\EndIf
	\end{algorithmic}
\end{algorithm}
\begin{theorem}\label{thm:mm}
	Given a simple graph $G=(V,E)$ on $n$ vertices and $m$ edges, we can in $\CL\P$ find the size of the maximum matching, $\nu(G)$ of $G$.
\end{theorem}
\begin{proof}
	Let $x_1,\ldots,x_m$ be edge variables. Consider the Tutte matrix $T$, given by 
\[
T_{uv} = \begin{cases}
	\; x_{i} & u<v,\ i=\{u,v\}\in E,\\[2mm]
	\; -x_{i} & u>v,\ i=\{u,v\}\in E,\\[2mm]
	\; 0 & \text{otherwise.}
\end{cases}
\] 
Assume that the catalytic tape is divided into blocks $(\mcC_1,\ldots,\mcC_N,\mcB)$ where $\mcC_r=(a_1,\ldots,a_m)\forall r\in [N]$ with each $a_i\in[n^{10}]$ is $10\log n$ bits long, and $\mcB$ is an extra block which shall be used for rank computations; $N$ is $n^3$. The algorithm proceeds as follows: We iterate over all $\mcC_1,\ldots,\mcC_N$, suppose we are processing $\mcC_r=(a_1,\ldots,a_m)$, then we follow the following steps (as described in \Cref{alg2}):
\begin{enumerate}
	\item Let $T'$ be the evaluation of $T$ given by $\mcC_r$. We execute \Cref{alg2} which either reports $\nu(G)$ or not. If we get $\nu(G)$ from this call, then we proceed to call \Cref{alg3} on all previous catalytic blocks, i.e. jump to Step $4$ with the temporary variable $index\gets r-1$. Now suppose that \Cref{alg2} does not return $\nu(G)$. Therefore, from \Cref{geelen}, the only two possibilities are:
	\begin{enumerate}
		\item[$(a)$] $\exists i\in E, a\in [n^{10}]$ such that $\rank T'_{x_i\gets a}=\rank T'+2$.
		\item[$(b)$] $\exists i\in E, a\in [n^{10}]$ such that $\rank T'_{x_i\gets a}=\rank T'\wedge D(T')\subsetneq D(T'_{x_i\gets a})$.
	\end{enumerate}Consider the two cases $2A,2B$ that deals with $(a)$ and $(b)$ respectively, as follows--
	\item[$2A$.] In this case, after finding an $i$ as in $(a)$, from \Cref{2A}, we know that for no other substitution $x_i\gets a$ and $a\ne a_i$, is $\rank T'_{x_i\gets a}=\rank T'$. Therefore, we replace $\mcC_r$ with $(2A,i,\rank T',a_1,\ldots,a_{i-1},a_{i+1},\ldots,a_m)$. This saves $10m\log n-\mcO(1)-\log m-\log n-10(m-1)\log n> 6\log n$ (here $\mcO(1)$ comes from writing $2A$).
	\item[$2B$.] Again, first we find an index $i$, and value $a$ witnessing case $(b)$. Now, we pick a $u\in D(T'_{x_i\gets a})\setminus D(T')$. Again from \Cref{2B} we know that $a_i$ is the unique value for which $D(T')$ does not contain $u$. Hence we replace $\mcC_r$ with $(2B,i,u,a_1,\ldots,a_{i-1},a,a_{i+1},\ldots,a_m)$. This again saves $10m\log n-\mcO(1)-\log m-\log n-10(m-1)\log n>6\log n$ bits.
	\item[$3.$] Suppose we just processed $\mcC_N$ and still did not find $\nu(G)$. In this case, we have freed up at least $6N\log n=\mcO(n^3\log n)$ space for $N=n^3$. We shift the catalytic strings to make this space contiguous, and run any standard polynomial time algorithm for matching and find $\nu(G)$. Next set $index=N$ and go to step $4$.
	\item[$4.$] We have either reached this step if we have processed all catalytic blocks, then $index=N$, and otherwise we have jumped from step $1$ with some value of $index$. In either case, we now recompute the modified catalytic strings $\mcC_1,\ldots,\mcC_{index}$. For each $r\in[index]$, we execute \Cref{alg3} on $\mcC_r=(case,j,\beta,a_1,\ldots,a_{j-1},a_{j+1},\ldots,a_m)$. Here is how it works: 
	\begin{itemize}
		\item[$(a)$] This is when $case=2A$. Therefore, we have $k=\beta$ to be the rank of $T'$ with the evaluation of $T$ given by the original catalytic tape (before we made any modified any catalytic bits). Therefore, we find the unique $a$ such that $\rank T'_{x_j\gets a}=k$, and restore the catalytic tape by setting $a_j\gets a$. The correctness is again given by \Cref{2A}.
		\item[$(b)$] This time $case=2B$. Therefore $u=\beta$. Again find the unique $a$ such that $u\notin D(T'_{x_j\gets a})$ and restore $a_j\gets a$. The correctness follows from \Cref{2B}.
	\end{itemize}
\end{enumerate}
The above algorithm clearly runs in polynomial time and the catalytic space used is $\mcO(mn^2\log n+|\mcB|)=\poly(n)$ with work space $\mcO(\log n)$, and the catalytic tape is always restored when the algorithm halts. Thus the above gives the required $\CL\P$ algorithm to compute $\nu(G)$.
\end{proof}
\begin{corollary}\label{gallai}
	Given a simple graph $G=(V,E)$, we can in $\CL\P$ compute the Gallai-Edmonds Decomposition of $G$.
\end{corollary}
\begin{proof}
	Notice that in the above algorithm, if we succeed in finding a `good' catalytic tape (in \Cref{alg2}), then we have also found $D(T)$ which is equal to $D(G)$. In case we do not find a `good' catalytic tape, we have saved enough free space to compute $D(G)$ from scratch using \cite{Geelen2000}. After computing $D(G)$ it is easy to get $A(G)$ and $C(G)$ as $A=(V\setminus D)\cap N_G[D], C=V\setminus(A\cup D)$ where $N_G[D]$ is the set of neighbors of $D(G)$ in $G$.
\end{proof}
\begin{corollary}\label{trank}
	Given a simple graph $G$ and its Tutte matrix $T$, in $\CL\P$ we can find an assignment $T'$ of $T$ that maximizes the rank of $T$.
\end{corollary}
\begin{proof}
	In the above algorithm, if we find a good catalytic tape via \Cref{alg2}, then the catalytic bits themselves give an evaluation $T'$ that is of the highest possible rank. In case we have not found any `good' catalytic tapes, we have freed up enough space to run Geelen's algorithm \cite{Geelen2000} directly.
\end{proof}
\begin{remark}\label{rem:mm}
	Notice that from \Cref{trank} we get a $\CL$ algorithm that finds an evaluation $T'$ of $T$ in the following sense:
	\begin{itemize}
		\item Either the program saves enough space, and finds $T'$, and can later restore the catalytic tape whenever necessary.
		\item Or, it finds the evaluation $T'$ from a given string (read as assignments to the variables) on the catalytic tape.
	\end{itemize}
	This shall be used implicitly when we prove \Cref{pms,thm:smm}.
\end{remark}

\section{Hunt for a Maximum Matching}
Finally we show how to search for a maximum matching in $\CL$. We break it down into two steps, the second step being a standard reduction, but the first step is fairly technical and uses new ideas:
\begin{enumerate}
    \item First, we show how to search for a perfect matching in a general graphs, given that the graph contains a perfect matching. We shall crucially use the fact that we can find an assignment to the Tutte matrix that achieves full rank, \Cref{trank}.
    \item Next we give a procedure to find a maximum matching in general graphs, by using $(1)$.
\end{enumerate}
\subsection{Perfect Matching Search in general graphs}
Throughout this section, we assume that we are given a simple graph $G=(V,E)$, that has a perfect matching.

\paragraph*{Notation}

For a weight assignment $W:E\to \bbN$ on the edges and an evaluated Tutte matrix $T'$, define $T'\circ W$ as the skew-symmetric matrix:
\[
(T'\circ W)_{uv} = \begin{cases}
	\; T'_{uv}\cdot z^{W(uv)} & \{u,v\}\in E,\\[2mm]
	\; 0 & \text{otherwise.}
\end{cases}
\]

Let $T'$ be an evaluation of $T$ such that $\rank T'=n$, and let $W$ be a weight assignment on the edges.
Denote $A=T'\circ W$, and for any edge $e=(i,j)$, denote $A_e$ as the matrix obtained from $A$ by multiplying $A_{ij},A_{ji}$ with the variable $y$ and setting $W(e)=0$, i.e. $(A_e)_{ij},(A_e)_{ji}$ are $T'_{ij}y$ and $T'_{ji}y$ respectively.

Let $\mcM$ be the set of all perfect matchings of $G$. For a perfect matching $M\in\mcM$ in $G$, denote by $c_M$ the coefficient\footnote{i.e. Let $T'_X$ be the same matrix as $T'$ but with $T'_{ij},T'_{ji}$ multiplied by $x_{ij}$ if $(i,j)\in E$. Now let $x_M=\Pi_{e\in M}x_e$. Then $c_M$ is the coefficient of $x_M$ in $\Pf(T'_X)$.} of the monomial corresponding to $M$ in $\Pf (T')$. 
Therefore, we have 
\[
\Pf(A)=\sum_{M\in\mcM}c_M\cdot z^{W(M)}
\]
Observe that we can write $\Pf(A_e)=P_0(z)+yP_1(z)$ where $P_0,P_1$ are univariate polynomials in $z$.
Furthermore, for $w\in \bbN$, let $$c_w=\sum_{M:W(M)=w}c_M,$$ and for $e\in E$: $$c_{w}^{\bar e}=\sum_{M:W(M)=w,e\notin M}c_M\text{ and }c_w^e=c_w-c_w^{\bar e}=\sum_{M:W(M)=w,e\in M}c_M.$$ Therefore, we have that:
$$
P_0(z)=\sum_{w\in\bbN}c_w^{\bar e}z^w, P_1(z)=\sum_{w\in\bbN}c_{w+W(e)}^ez^w
$$

We know that $T'$ is full rank, hence $\Pf (A)$ is a non-zero polynomial in $z$, since at $z=1$, $\Pf(A)|_{z=1}=\Pf(T')\ne 0$. Let $c_{w_0}z^{w_0}$ be the minimum degree term of $\Pf (A)$ with the non-zero coefficient $c_{w_0}$. For an edge $e\in E$, let $c_{w_{\bar e}}^{\bar e} z^{w_{\bar e}}$  denote the minimum degree monomial in $P_0$ not containing the variable $y$ (if such a monomial exists, i.e. if $P_0$ is a non-zero polynomial in $z$), and $c_{w_e+W(e)}^e z^{w_e}$ be the minimum degree monomial containing $y$ in $P_1$ (if it exists, i.e. if $P_1$ is non-zero).

\paragraph*{A high level overview}
 Notice that $w_0$ may not be the weight of the min-weight perfect matching in $G$, since the corresponding term for the minimum weight matching may cancel out in $\Pf(A)$. But, nevertheless, we attempt to isolate a matching of weight ${w_0}$, though not in the conventional sense i.e. there may be multiple matchings with this weight. Notice that if $W$ isolates a unique minimum weight perfect matching of $G$ with total weight $w^\ast$, then the set $\{e\in E:c_{w^\ast}^{\bar e}=0\}$ is precisely this unique perfect matching. But, in our context there might not be a unique $w_0$ weight matching, but nevertheless, we show that either of the two holds:
 \begin{enumerate}
     \item the set $\{e\in E:c_{w_0}^{\bar e}=0\}$ is a perfect matching of weight $w_0$.
     \item We can find an edge $e\in E$, such that $W(e)=w_{\bar e}-w_e$, and using this we shall develop a `compress or compute' algorithm.
 \end{enumerate}
 We make these notions formal in the next lemma. The following is the crucial lemma that helps us find a perfect matching in $\CLP$. can

\begin{lemma}\label{unique}
	With the above notation, let $w_0$ be the least degree such that $z^{w_0}$ is present with non-zero coefficient in $\Pf(A)$. Then, one of the following holds:
	\begin{enumerate}
		\item $\exists e\in E$ such that $W(e)=w_{\bar e}-w_e$ where $z^{w_{\bar e}}$ is the least degree monomial with non-zero coefficient in $P_0$ and $z^{w_{ e}}$ is the least degree monomial with non-zero coefficient in $P_1$ and $\Pf(A_e)=P_0+yP_1$.
		\item $M=\{e\in E: c_{w_0}^{\bar e}=0\}$ is a perfect matching. Moreover $M$ is exactly the set of edges $e\in E$ such that $P_0$ does not contain the monomial $z^{w_0}$.
	\end{enumerate}
\end{lemma}
\begin{proof}
 Notice that $c_{w_0}\ne 0$ and $\forall w<w_0, c_{w}=0$. We divide the proof into the following two mutually exclusive cases:

 \texttt{Case 1.} $\exists e\in E$, such that both $c_{w_0}^e\ne 0$ and $c_{w_0}^{\bar e}\ne 0$.

 \texttt{Case 2.} $\forall e\in E$, at least one of $c_{w_0}^e$ and $c_{w_0}^{\bar e}$ is $0$.

 We will show that if Case $1$ holds, then for the same edge $e$ such that $c_{w_0}^e\ne 0$ and $c_{w_0}^{\bar e}\ne 0$ holds, we have the equality $W(e)=w_{\bar e}-w_e$. And if Case $2$ holds, we shall show that $M=\{e\in E: c_{w_0}^{\bar e}=0\}$ is a perfect matching in $G$. We start with the first case.

 \texttt{Case 1.} 
 Suppose that there exists an edge $e\in E$ such that $c_{w_0}^e$ and $c_{w_0}^{\bar{e}}$ both are nonzero. Write as before

$$
\Pf\left(A_e\right)=P_0(z)+y P_1(z)
$$

and $w_{\bar{e}}, w_{e}$ are the minimum degrees in $P_0, P_1$ respectively. Since $c_{w_0}^e$ and $c_{w_0}^{\bar{e}}$ both are nonzero, we have $w_{\bar e},w_e-W(e)\le w_0$. If $w_{\bar{e}}=w_{e}+W(e)$ then we are done. Else, $w_{\bar{e}}<w_e+W(e)$ or $w_{\bar{e}}>w_e+W(e)$ are the two cases we consider:
\begin{enumerate}
    \item In the first case, we have $w_{\bar e}<w_e+W(e)$, and therefore the coefficient of $z^{w_{\bar e}-W(e)}$ is zero in $P_1$, i.e, $c_{w_{\bar e}}^e=0$. Thus, observe that $c_{w_{\bar e}}=c_{w_{\bar e}}^e+c_{w_{\bar e}}^{\bar e}\ne 0$, implying that $w_{\bar e}\ge w_0$. Combining with the fact that $w_{\bar e}\le w_0$ we have $w_{\bar e}=w_0$. But then we have $c_{w_0}^e=0$ that contradicts our assumption on $e$.
    \item The second case is symmetric, with the assumption $w_{\bar{e}}>w_e+W(e)$. We again see that the coefficient of $z^{w_e+W(e)}$ in $P_0$ is zero, i.e, $c_{w_{ e}+W(e)}^{\bar e}=0$. This in turn means $c_{w_e+W(e)}\ne 0$, implying $w_e+W(e)\ge w_0$. Combining with the fact that $w_e+W(e)\le w_0$, we have the equality $w_e+W(e)=w_0$. But then we have $c_{w_0}^{\bar e}=0$ giving a contradiction.
\end{enumerate}
 We have therefore shown that $W(e)=w_{\bar e}-w_e$. Now we move on to the other case.

    \texttt{Case 2.} We now assume that $\forall e\in E$, at least one of $c_{w_0}^e$ and $c_{w_0}^{\bar e}$ is $0$. Let $M=\{e\in E:c_{w_0}^{\bar e}=0\}$. Since $c_{w_0}\ne 0$, we have $M= \{e\in E: c_{w_0}^e\ne 0\wedge c_{w_0}^{\bar e}=0\}$. We wish to show that $M$ is a perfect matching of $G$. We proceed by contradiction. Consider the graph $G'=(V, M)$. If $M$ is not a perfect matching, there is either a vertex with degree $0$ in $G'$, or some vertex has degree at least $2$ in $G'$. We divide the proof into these two cases and show a contradiction in both of them:
    \begin{enumerate}
        \item Suppose $v$ has degree $0$ in $G'$. Let $e_1,\ldots,e_k$ be the edges in $G$ that are incident on $v$. For all $e_i$ we know that $e_i\notin M\implies c_{w_0}^{ e_i}= 0$. But then we have 
        $$c_{w_0}=\sum_{i\in[k]}c_{w_0}^{e_i}=0$$ which is a contradiction, since $c_{w_0}\ne 0$.
        \item Now suppose a vertex $v$ has degree $\ell$ in $G'$ where $\ell>1$. Suppose $e_1,\ldots,e_k$ are the edges incident on  $v$ in $G$ and $e_1,\ldots,e_\ell\in M, e_{\ell+1},\ldots,e_k\notin M$. Therefore, we have the following:
        $$\forall i\in[\ell]: c_{w_0}^{\bar e_i}=0, \forall j>\ell: c_{w_0}^{e_j}= 0 $$
        and therefore, $$\forall i\in[\ell]: c_{w_0}^{\bar e_i}=\sum_{j\in[\ell]\setminus\{i\}}c_{w_0}^{e_j}=0.$$
        Summing the above for all $i\in[\ell]$, we get 
        $$\sum_{i\in[\ell]}c_{w_0}^{\bar e_i}=\sum_{i\in[\ell]}\sum_{j\in[\ell]\setminus\{i\}}c_{w_0}^{e_j}=(\ell-1)\sum_{i\in[\ell]}{c_{w_0}^{e_i}}=0.$$
 But since $\ell>1$, we have: $$c_{w_0}=\sum_{i\in[\ell]}c_{w_0}^{e_i}=0.$$
 Again this gives a contradiction since $c_{w_0}\ne 0$.
        
    \end{enumerate}
    
By the above argument, we have that every vertex has degree exactly $1$ in $G'$. This happens precisely when $M$ is a perfect matching of $G$, and we are done.

\end{proof}

Finally, we come to the following algorithm:

\begin{lemma}\label{pms}
	There is a $\CL\P$ procedure that, given a simple undirected graph $G$ on $n$ vertices and $m$ edges which contains a perfect matching, returns a perfect matching of $G$.
\end{lemma}
\begin{proof}
Assume that the catalytic tape is given as $n^3$ weight assignments $(W_1,W_2,\ldots,W_{n^3})$ where each $W_r:E\to [n^{10}]$ is given as $(W_i(e_1),\ldots,W_i(e_m))$. We assume that we are given an assignment to the Tutte matrix $T'$ by \Cref{trank} in the sense of \Cref{rem:mm}. Then we do the following, by processing $W_1,\ldots,W_{n^3}$ sequentially:
	\begin{enumerate}
		\item Suppose, we are processing $W_r$. Set $A=T'\circ W_r$. For all $e\in E$ check whether $W_r(e)=w_{\bar e}-w_e$ where $z^{w_{\bar e}}$ is the least degree monomial with non-zero coefficient (can be found using interpolating) in $P_0$ and $z^{w_{ e}}$ is the least degree monomial with non-zero coefficient in $P_1$. (recall that $A_e$ is the matrix obtained from $A$ by multiplying $A_{ij},A_{ji}$ by the variable $y$ and setting $W(e)=0$, and $\Pf(A_e)=P_0+yP_1$) If we find such an edge $e_i$ (i.e. the $i^{th}$ edge), then reset the current catalytic tape string to $(i,W_r(e_1),\ldots,W_r(e_{i-1}),W_r(e_{i+1}),\ldots,e_m)$. This frees up $10\log n-\log m\ge8\log n$ bits of space. On the other hand, if we did not find such an edge, using \Cref{unique}, we find a perfect matching as follows:
		
		Compute $w_0$ to be the minimum degree such that $z^{w_0}$ is present in $\Pf(A)$. Find all edges $e$ such that $P_0$ does not contain the monomial $z^{w_0}$ where $\Pf(A_e)=P_0+yP_1$. This set of edges precisely forms the perfect matching $M=\{e\in E: c_{w_0}^{\bar e}=0\}$ by \Cref{unique}. After finding this matching and returning it, we move to step $3$ (to restore the previous catalytic strings) with the temporary variable $index\gets r-1$.
		\item Suppose that we just processed $W_{n^3}$ and still did not find a perfect matching. In this case, we have saved up $\mcO(n^3\log n)$ bits of space, where we can run any standard algorithm (say \cite{Geelen2000} with the self-reduction step) to search for a perfect matching. Then we set $index\gets n^3$ and go to step $3$ to restore the catalytic strings.
		\item We restore all the weights $W_1,\ldots,W_{index}$ as follows: Suppose we are restoring $W$. We are given the catalytic string $(i,W(e_1),\ldots,W(e_{i-1}),W(e_{i+1}),\ldots,W(e_m))$. Compute $A_{e_i}$ (can be done since we set $W(e_i)=0$). Compute $w_{\bar e_i},w_{e_i}$ such that $z^{w_{\bar e_i}}$ is the least degree monomial with non-zero coefficient in $P_0$ and $z^{w_{ e_i}}$ is the least degree monomial with non-zero coefficient in $P_1$. Set $W(e_i)=w_{\bar e_i}-w_{e_i}$. Reset $W$ to $(W(e_1),\ldots,W(e_m))$.
	\end{enumerate}
	Clearly we have always restored the catalytic tape to its initial state when our computation halts, and we always return a perfect matching. The space bounds are: $\mcO(\log n)$ workspace, $\mcO(n^3m\log n)$ catalytic space, and the algorithm runs in polynomial time.
\end{proof}

\subsection{Maximum Matching Search}
Finally, we have the main result of this section, as follows:
\begin{theorem}\label{thm:smm}
	Given a simple graph $G=(V,E)$, we can construct a maximum matching $M$ of cardinality $\nu(G)$ in $\CL\P$.
\end{theorem}
\begin{proof}
	Let \(G=(V,E)\) be a graph, and suppose we are given the size \(k=\nu(G)\) of a maximum matching in \(G\) by \Cref{thm:mm}. 
Let $d=n-2k$.

Construct a graph \(G'=(V',E')\) as follows:

Add \(d\) new dummy vertices
$
w_1,\dots,w_d 
$. Let
$
V' = V \cup \{w_1,\dots,w_d\}
$,
and
$
E' = E \cup \{(w_i,v) : i\in[d],\, v\in V\}.
$

Clearly, there perfect matching in $G'$.
Suppose \(M\) is a matching of size \(k\) in \(G\). Then exactly
$
|V|-2k = d
$ vertices of \(V\) are unmatched by \(M\). Match each dummy vertex \(w_i\) to a distinct unmatched vertex of \(V\). Together with \(M\), this gives a perfect matching of \(G'\).

Conversely, suppose \(G'\) has a perfect matching \(M'\). Since there are exactly \(d\) dummy vertices, and each dummy vertex can only contribute one matched edge, exactly \(d\) edges of \(M'\) are incident to dummy vertices (since no edge is between two dummy vertices). Removing all such edges leaves a maximum matching $M$ in \(G\).

Therefore, we can find a maximum matching of $G$ in $\CLP$ as follows: First construct $G'$ in logspace, and compute a perfect matching $M'$ of $G'$ using \Cref{pms}. Next, using the above stated logspace procedure, compute a maximum $M$ matching of $G$ by removing the edges of $M'$ incident on the dummy vertices.

\end{proof}







\section{Rank of Mixed Matrices}\label{sec:mixed}
In this section, we prove \Cref{thm:main2}. 
\subsection{Mixed matrices}\label{sec:mixedmatrix}
 
\begin{definition}\label{def:mix}
	Let $A$ be an $r\times c$ matrix where each entry of $A$ is either an indeterminate, or an element of $\bbQ$, where the set of variables is $X=\{x_{ij}:i\in[r],j\in[c],A_{ij}=x_{ij}\}$.
\end{definition}
Mixed matrices can also be written as $A=K+Q$ where $Q$ is a scalar matrix over $\bbQ$ and $K$ is an indeterminate matrix (where no two indeterminates in two coordinates are the same, in the same sense as \Cref{def:mix}).
For any $r\times c$ matrix $A$ whose rows are indexed by $L_r$ and columns by $L_c$, define $D(A)=\{i\in L_r\cup L_c: \rank (A\setminus \{i\})=\rank A\}$. Let $E=\{(i,j):A_{ij}\in X\}$, further index these tuples such that $E=\{i:i=(i',j'),A_{i'j'}\in X\}$.

\begin{theorem}[Geelen, \cite{Geelen1999}]\label{geelen2}
	Let $A$ be a mixed matrix over the variable set $X$ where $|X|=m$, and let $A'$ be an evaluation of $A$. Then either $\rank A'=\rank A$ and $D(A')=D(A)$, or one of the following two statements hold:
	\begin{itemize}
		\item $\exists i\in E, a\in [n^{10}]$ such that $\rank A'_{x_{i}\gets a}=\rank A'+1$.
		\item $\exists i\in E, a\in [n^{10}]$ such that $\rank A'_{x_{i}\gets a}=\rank A'\wedge D(A')\subsetneq D(A'_{x_{i}\gets a})$.
	\end{itemize}
\end{theorem}
Analogously to \Cref{2A,2B} ($2A,2B$), we show the following uniqueness lemmata, which will be crucial for our $\CL$ algorithm to work.
\begin{lemma}[$2A'$]\label{2A'}
	Let $A'$ be an evaluation of $A$ and $i\in E, a\in [n^{10}]$ such that $\rank A'_{x_i\gets a}=\rank A'+1$, and in $A'$, $x_i$ is $a_i$. Then $\forall \tilde a\ne a_i$, we have $\rank A'_{x_i\gets \tilde a}=\rank A'+1$.
\end{lemma}
\begin{proof}
	Let $k=\rank A'$. 
	Let $X\subseteq L_r,Y\subseteq L_c, |X|=|Y|=k$ be a minor of $A$ that is non-singular in $A'_{x_i\gets a}$. Notice that $L(y)=\Det (A'_{x_i\gets y}[X,Y])$ is linear in $y$. Since $L(a)\ne 0, L(a_i)=0$, clearly $\forall \tilde a\ne a_i, L(\tilde a)\ne 0$ i.e $A'_{x_i\gets \tilde a}[X,Y]$ is non-singular.
	Furthermore, on changing one entries of $A'$ (i.e. by changing the variable $x_i$) the rank can go up by atmost $1$. Therefore, $\rank A'_{x_i\gets \tilde a}=\rank A'+1$ holds $\forall \tilde a\ne a$.
\end{proof}

\begin{lemma}[$2B'$]\label{2B'}
	Let $A'$ be an evaluation of $A$ and $i\in E, a\in [n^{10}]$ such that $\rank A'_{x_i\gets a}=\rank A'=k$ where in $A'$, $x_i$ is set to $a_i$. Suppose $\exists u\in D(A'_{x_i\gets a})\setminus D(A')$. Then $\forall \tilde a\ne a_i$, we have $u\in D(A'_{x_i\gets \tilde a})$. 
\end{lemma}
\begin{proof}
	From the given assumptions, we have a $X_u\subseteq L_r,Y_u\subseteq L_c ,|X_u|=|Y_u|=k$ such that $A'_{x_i\gets a}[X_u,Y_u]$ is non-singular and $u\notin X_u\cup Y_u$. Since $u\notin D(A')$, $A'[X_u,Y_u]$ is singular.
	\begin{claim}
		$\forall \tilde a\ne a_i$, $\rank A'_{x_i\gets \tilde a}\le k$.
	\end{claim}
	\begin{proof}
		Suppose $\rank A'_{x_i\gets \tilde a}>k$ for some $\tilde a$. Therefore, $\tilde a\notin\{ a,a_i\}$. Therefore, we can find $X,Y$ as a square minor of size larger than $k$ such that $A'_{x_i\gets \tilde a}[X,Y]$ is non-singular. But then $\wtL (y)=\Det(A'_{x_i\gets y}[X,Y])$ is linear and $\wtL(a)=\wtL(a_i)=0$. Therefore, $\wtL(\tilde a)=0$ gives a contradiction.
	\end{proof}
	Again $L(y)=\Det (A'_{x_i\gets y}[X_u,Y_u])$ is linear in $y$, such that $L(a)\ne 0, L(a_i)=0$.
	Therefore, $\forall \tilde a\ne a_i$ we have $L(\tilde a)\ne 0$ i.e $A'_{x_i\gets \tilde a}[X_u]$ is non-singular and by the above claim $X_u,Y_u$ is of maximum possible size (i.e. $\rank A'_{x_i\gets \tilde a}=|X_u|=|Y_u|=k$), thus $u\in D(A'_{x_i\gets \tilde a})$.
\end{proof}
\begin{lemma}\label{rankc}
Suppose that we are given an evaluation of a mixed matrix $A$, namely $A'$. Then we can check, in $\CL$ if $\rank A=\rank A'$ or not.
\end{lemma}
\begin{proof}
	The proof is very similar to \Cref{alg1} and uses \Cref{geelen2}. For all $i$ and all perturbations $x_i\gets a (\forall a\in[m^{10}])$ check if the rank increases by $1$ or if the $D$ set increases. If for no such perturbation this happens, we conclude by \Cref{geelen2} that $\rank A'=\rank A$, and otherwise  we know that $\rank A'<\rank A$.
\end{proof}

\subsection{A catalytic rank algorithm}

In this section, we show that the rank of mixed matrices can be computed in $\CL$. Let $A$ be a mixed matrix with variables $x_1,\ldots,x_m$. The algorithm is essentially the same as the matching algorithm described in \Cref{sec:matching}. We sample an assignment $a_1,\ldots,a_m$ for these variables from the catalytic tape. Either this assignment maximizes the rank of the evaluated matrix $A'$, or it does not. We can verify this using \Cref{rankc}. If the rank is indeed maximized, we proceed to restore the modified catalytic tape.

Otherwise, if the rank is not maximized, we consider two cases:
\begin{itemize}
    \item[(a)] Suppose that by perturbing one of the variables $x_i$, the rank of $A'$ can be increased. In this case, by \Cref{2A'}, for all values other than $a_i$, the rank increases. Therefore, we discard $a_i$ from the catalytic tape and store the index $i$ together with the rank. Later, we can recover $a_i$ as the unique value for which the rank remains unchanged.
    
    \item[(b)] Otherwise, perturbing some variable $x_i$ does not increase the rank, but enlarges the set $D(A')$, i.e., there exists a vertex $u$ that gets added to $D(A')$. By \Cref{2B'}, $a_i$ is the unique value for which $u \notin D(A')$. Thus, we again discard $a_i$ and store its index $i$ together with $u$. As before, $a_i$ can be recovered when needed using this characterization.
\end{itemize}

Finally, if the sampled assignment already yields the maximum rank, we are done. Otherwise, we have freed enough space to recompute the rank from scratch using \cite{Geelen1999}. A formal proof follows.

\begin{theorem}\label{thm:rank}
	Let $A$ be an $r\times c$ dimensional mixed matrix on variables $X$ indexed over the set $E$ of size $m\le n^2$, and $n=r+c$. Then, there is a $\CL\P$ procedure that finds $\rank A$, and an assignment to $X$ (and thus an evaluation of $A$, namely $A'$) that is of the maximum possible rank, i.e., $\rank A'=\rank A$.
\end{theorem}
\begin{proof}
	The algorithm mirrors the proof of \Cref{thm:mm}.
	We assume that the catalytic tape is divided into blocks $(\mcC_1,\ldots,\mcC_N,\mcB)$ where $\mcC_t=(a_1,\ldots,a_m)\forall t\in [N]$ with each $a_i\in[n^{10}]$ is $10\log n$ bits long, and $\mcB$ is for rank computations; $N=n^3$. The algorithm proceeds as follows: We iterate over all $\mcC_1,\ldots,\mcC_N$ as follows, suppose we are processing $\mcC_t$:
	\begin{enumerate}
		\item Let $A'$ be the evaluation of $A$ given by $\mcC_t$ i.e. by putting $x_i\gets a_i\forall i\in[m]$. We first check if $\rank A'=\rank A$ or not, using \Cref{rankc}. If yes, then we have found a required assignment, and we proceed to restore all previous catalytic blocks, i.e. jump to Step $4$ with the temporary variable $index\gets t-1$. Now suppose we did not get the required assignment. Therefore, from \Cref{geelen2}, the only two possibilities are:
		\begin{enumerate}
			\item[$(a)$] $\exists i\in E, a\in [n^{10}]$ such that $\rank A'_{x_i\gets a}=\rank A'+1$.
			\item[$(b)$] $\exists i\in E, a\in [n^{10}]$ such that $\rank A'_{x_i\gets a}=\rank A'\wedge D(A')\subsetneq D(A'_{x_i\gets a})$.
		\end{enumerate}Consider the two cases $2A',2B'$ that deals with $(a)$ and $(b)$ respectively, as follows--
		\item[$2A'$.] In this case, after finding an $i$ (we enumerate over all $i$ to find such an $i$) as in $(a)$, from \Cref{2A'}, we know that for all substitutions $x_i\gets a$ and $a\ne a_i$, we have $\rank A'_{x_i\gets a}\ne\rank A'$. Therefore, we replace $\mcC_t$ with $(2A',i,\rank A',a_1,\ldots,a_{i-1},a_{i+1},\ldots,a_m)$. This saves $ 6\log n$.
		\item[$2B'$.] Again, first we find an index $i$, and value $a$ witnessing case $(b)$. Now, we pick a $u\in D(A'_{x_i\gets a})\setminus D(A')$. Again from \Cref{2B'} we know that $a_i$ is the unique value for which $D(A')$ does not contain $u$. Hence we replace $\mcC_r$ with $(2B',i,u,a_1,\ldots,a_{i-1},a,a_{i+1},\ldots,a_m)$. This again saves $6\log n$ bits.
		\item[$3.$] Suppose we just processed $\mcC_N$ and still did not find a required assignment. In this case, we have freed up at least $\mcO(n^3\log n)$ space for $N=n^3$. We shift the catalytic strings to make this space contiguous, and run Geelen's algorithm to find an assignment that maximizes the rank. Next, set $index=N$ and go to step $4$.
		\item[$4.$] In either case, we now recompute the modified catalytic strings $\mcC_1,\ldots,\mcC_{index}$. For each $t\in[index]$, we execute the following with $\mcC_t=(case,j,\beta,a_1,\ldots,a_{j-1},a_{j+1},\ldots,a_m)$:
		\begin{itemize}
			\item[$(a)$] This is when $case=2A'$. Therefore, we have $k=\beta$ to be the rank of $A'$ with the evaluation of $T$ given by the original catalytic tape. Thus, we find the unique $a$ such that $\rank A'_{x_j\gets a}=k$, and restore $a_j\gets a$. The correctness is again given by \Cref{2A'}.
			\item[$(b)$] This time $case=2B'$. Therefore $u=\beta$. Again find the unique $a$ such that $u\notin D(T'_{x_j\gets a})$ and restore $a_j\gets a$. The correctness follows from \Cref{2B'}.
		\end{itemize}
		Finally we reset  $\mcC_t\gets (a_1,\ldots,a_{j-1},a_j,a_{j+1},\ldots,a_m)$.
	\end{enumerate}
	The above algorithm clearly runs in polynomial time and the catalytic space used is $\mcO(mn^2\log n+|\mcB|)=\poly(n)$ with work space $\mcO(\log n)$, and the catalytic tape is always restored when the algorithm halts. Thus the above gives the required $\CL\P$ algorithm that finds $\rank A$ and an assignment $\vec a$ such that $\rank A=\rank A_{x_i\gets a_i}$.
\end{proof}

\begin{corollary}
	Given two linear matroids $M_1$ and $M_2$, we can find the cardinality of the maximum common independent set in $\CL\P$.
\end{corollary}
\begin{proof}
	Let $M_1$ and $M_2$ be matroids on a common ground set $E$, $|E|=n$. Suppose $M_1$ and $M_2$ have their representation matrices $A_1$ and $A_2$ repectively. Let $\{x_1\ldots,x_n\}$ be the set of variables. Consider the following matrix:
	
	$$
	A:=\left(\begin{array}{c|ccc}
		0 & & A_1 & \\
		\hline & x_1 & & \\
		A_2^T & & \ddots & \\
		& & & x_n
	\end{array}\right) .
	$$

	In \cite{Murota1995Mixed,murota2000matrices}, Murota showed that the size of a maximum common independent set of $M_1$ and $M_2$ is $\rank A-n$. Since $A$ is a mixed matrix, and we can compute $\rank A$ using \Cref{thm:rank}, and we are done.
\end{proof}

\section{$(1-\eps)$-approximation for Edmonds's Problem}
In this section we prove \Cref{thm:main3}.

Given a matrix space $\mcA=\langle A_1,\ldots , A_m\rangle $ where each matrix is a $n\times n$ matrix over some underlying field $\bbF$ of size greater that some polynomial in $n$, say $|\bbF|>n^{c}$ for a carefully chosen $c$ that depends inversely on $\eps$. We want to approximate the maximum rank that any matrix in $\mcA$ can attain. Consider the indeterminates $x_1,\ldots,x_m$, we want to approximate the rank of $\sum_{i\in[m]}A_ix_i=:\rank(\mcA)$.

\begin{theorem}[\cite{BlaserJindalPandey2018}]\label{thm:ed}
	Let $\mcA=\langle A_1,\ldots,A_m\rangle\le \bbF^{n\times n}$, and $\eps\in(0,1)$ be a fixed parameter with $\ell=\ceil{\frac 1 \eps-1}$. Let $a_i\in\bbF$ for $i\in[m]$ and $S$ be any subset of $\bbF$ of size $n^{c}$. Then either:
	\begin{enumerate}
		\item $\rank \Big(\sum_{i\in[m]}A_ia_i\Big)\ge (1-\eps)\cdot\rank \Big(\sum_{i\in[m]}A_ix_i\Big)$.
		\item or, $\exists I=\{i_1,\ldots,i_\ell\}\in\binom{[m]}\ell$ and $\lambda_{i_1}\ldots,\lambda_{i_\ell}\in S$ such that $$\rank \Big(\sum_{i\in[m]\setminus I}A_ia_i+\sum_{i\in I}A_i\lambda_i\Big)>\rank \Big(\sum_{i\in[m]}A_ia_i\Big)$$
	\end{enumerate}
\end{theorem}
Henceforth, we fix $S$ as the first $n^c$ elements of $\bbF$ in the canonical order. Even though we do not have any uniqueness lemma like \Cref{2A,2B}, we show the following which is essentially an application of the Schwartz-Zippel lemma, and shall be crucial for our purposes.
\begin{lemma}\label{bad}
	Suppose $\vec a=a_1,\ldots,a_m\in\bbF$ and $\exists I=\{i_1,\ldots,i_\ell\}\in\binom{[m]}\ell$ and $\lambda_{i_1}\ldots,\lambda_{i_\ell}\in S$ such that $$\rank \Big(\sum_{i\in[m]\setminus I}A_ia_i+\sum_{i\in I}A_i\lambda_i\Big)>\rank \Big(\sum_{i\in[m]}A_ia_i\Big).$$ Let $$\mfI_{\vec a,I}=\Big\{\lambda'_{i_1},\ldots,\lambda'_{i_\ell}\in S:\rank \Big(\sum_{i\in[m]\setminus I}A_ia_i+\sum_{i\in I}A_i\lambda'_i\Big)=\rank \Big(\sum_{i\in[m]}A_ia_i\Big)\Big\}.$$ Then $|\mfI_{\vec a,I}|\le n^{c\ell-c+1}$.
\end{lemma}
\begin{proof}
	Let $A'=\sum_{i\in[m]}A_ia_i, A''=\sum_{i\in[m]\setminus I}A_ia_i+\sum_{i\in I}A_i\lambda_i$.
	Let $X,Y\subseteq[n]$ be of maximum possible size such that $A''[X,Y]$ is a square non-singular matrix, but $A'[X,Y]$ is singular. Let $$p(x_{i_1},\ldots,x_{i_{\ell}})=\Det\Big(\sum_{i\in[m]\setminus I}A_ia_i[X,Y]+\sum_{i\in I}A_ix_i[X,Y]\Big)$$ be a polynomial in the indeterminates $x_{i_1},\ldots,x_{i_\ell}$. Clearly $p$ is of degree at most $n$ (to be precise, the degree is at most $|X|\le n$), and $p(\lambda_{i_1},\ldots,\lambda_{i_\ell})\ne 0$. Therefore, from \Cref{SZ}, the number of roots of $p$ in $S^\ell$ is upper bounded by $|S^\ell|.(n/|S|)=n^{c\ell-c+1}$. Moreover, for any $(\lambda'_{i_1},\ldots,\lambda'_{i_\ell})\in \mfI_{\vec a,I}$, we have $p(\lambda'_{i_1},\ldots,\lambda'_{i_\ell})= 0$. Therefore, $|\mfI_{\vec a,I}|\le n^{c\ell-c+1}$.
\end{proof}
We shall approximate $\rank(\mcA)$ upto a factor of $(1-\eps)$ in $\CL$. We sample $x_1\gets a_1,\ldots,x_m\gets a_m$ from the catalytic tape. Either we already have landed up at an approximately good assignment $A'$ in which case we begin reconstructing the modified catalytic strings. Otherwise, we know that we can find $I=(i_1,\ldots,i_\ell)$ and a perturbation to $x_{i_1},\ldots,x_{i_\ell}$ for which the rank increases. But then, the above lemma implies that $|\mfI_{\vec a,I}|$ is smaller than $|S|^\ell$ by a polynomially large factor. Notice that $I\in \mfI_{\vec a,I}$, and therefore $I$ is the $j^{th}$ tuple in $\mfI_{\vec a,I}$ for some $j$ that requires a small number of bits to represent. Therefore, we forget $a_{i_1},\ldots,a_{i_\ell}$ from $\vec a$ and remember instead $I$, $j$ and the rank of $A'$. We can later reconstruct $a_{i_1},\ldots,a_{i_\ell}$ by finding the $j^{th}$ tuple of $\mfI_{\vec a,I}$. Again, if we never find a suitable catalytic tape, then we have saved enough space to run \cite{BlaserJindalPandey2018} and find a $(1-\eps)$ approximation of the rank, and a required assignment. We now present the proof in detail:

\begin{algorithm}
	\caption{ProcessAssignment$(\mcC,\mcA)$}\label{alga1}
	\begin{algorithmic}[1]
		\State Given is $\mcC=(a_1,\ldots,a_m)$ where $a_i\in S\forall i\in[m]$.
		\State $A'\gets \sum_{i\in[m]}A_ia_i,k\gets \rank A'$
		\State $flag\gets 0$
		\For {$I=\{i_1,\ldots,i_\ell\}\in \binom{[m]}\ell, \lambda_{i_1},\ldots,\lambda_{i_\ell}\in S$}
		\State $A''\gets \sum_{i\in[m]\setminus I}A_ia_i+\sum_{i\in I}A_i\lambda_i$\Comment{$I=\{i_1,\ldots,i_\ell\}$}
		\If {$\rank A''>\rank A'$}
		\State $flag\gets 1$
		\EndIf
		\EndFor
		\If {$flag=0$}
		\State \Return $A'$
		\EndIf
		\State $j\gets 0$
		\For {$I=\{i_1,\ldots,i_\ell\}\in \binom{[m]}\ell, \lambda_{i_1},\ldots,\lambda_{i_\ell}\in S$}
		\State $A''\gets \sum_{i\in[m]\setminus I}A_ia_i+\sum_{i\in I}A_i\lambda_i$
		\If {$\rank A''=k$}
		\State $j\gets j+1$
		\EndIf
		\If{$(\lambda_{i_1}=a_{i_1})\wedge\ldots\wedge (\lambda_{i_\ell}=a_{i_\ell})$}
		\State Reset $\mcC\gets  (i_1\ldots,i_\ell,j,k,\vec a^{-I})$\Comment{$\vec a^{-I}$ is $\{a_1\ldots,a_m\}\setminus\{a_{i_1},\ldots,a_{i_\ell}\}$}
		\State \Return \Comment{Terminate the algorithm}
		\EndIf
		\EndFor
	\end{algorithmic}
\end{algorithm}

\begin{algorithm}
	\caption{RestoreAssignment$(\mcC,\mcA)$}\label{alga2}
	\begin{algorithmic}[1]
		\State Input: We are given $\mcC=(i_1\ldots,i_\ell,j,k,\vec a^{-I})$. 
		\State $j'=0$
		\For {$ \lambda_{i_1},\ldots,\lambda_{i_\ell}\in S$}
		\State $A''\gets \sum_{i\in[m]\setminus I}A_ia_i+\sum_{i\in I}A_i\lambda_i$
		\If {$\rank A''=k$}
		\State $j'\gets j'+1$
		\EndIf
		\If{$j=j'$}
		\State $a_{i_1}\gets\lambda_{i_1},\ldots,a_{i_\ell}\gets\lambda_{i_\ell}$
		\State Reset $\mcC\gets  (a_1,\ldots,a_m)$
		\State \Return \Comment{Terminate the algorithm}
		\EndIf
		\EndFor
		
	\end{algorithmic}
\end{algorithm}
\begin{theorem}\label{thm:edm}
	For any constant $\eps\in(0,1)$, and large enough field $\bbF$ such that $|\bbF|\ge n^{\mcO(1/\eps)}$, suppose we are given an instance of Edmond's problem, $\mcA=\langle A_1,\ldots , A_m\rangle $. Then, we have a $\CL\P$ algorithm that returns a matrix $A\in\mcA$ such that $\rank A\ge (1-\eps)\cdot \rank (\mcA)$.
\end{theorem}
\begin{proof}
	Let $S$ be the first $n^{c}$ elements of $\bbF$ that therefore require $c\log n$ bits to store. We shall throughout the algorithm, work with $S$. Fix $\ell=\ceil{1/\eps-1},N=n^3$ and $c=2\ell+3$. Assume that the catalytic tape is divided into blocks $(C_1,\ldots,C_N,\mcB)$ where $C_t=(a_1,\ldots,a_m)\forall t\in[N]$ and each $a_i\in S$ takes up $c\log n$ bits of storage. $\mcB$ is again for rank computations. Let $A=A_1x_1+\ldots A_mx_m$ where $x_1,\ldots ,x_m$ are variables. The algorithm proceeds as follows: We iterate over all $\mcC_1,\ldots,\mcC_N$ as follows, suppose that we are processing $\mcC_t=(a_1,\ldots,a_m)$:
	\begin{enumerate}
		\item We Execute \Cref{alga1}. Let $A'$ be the evaluation of $A$ given by $\mcC_t$ and $k$ be its rank i.e. by putting $x_i\gets a_i\forall i\in[m]$. If for no perturbation in $S$ for any $\ell$ variables, gives a matrix of larger rank than $A'$ then we have found $A'$ which approximates the rank of $A$ upto a factor of $(1-\eps)$ by \Cref{thm:ed}. Thus we jump to Step $4$ to reconstruct the previous catalytic blocks by setting $index=t-1$. On the other hand, suppose for some $I=(i_1,\ldots,i_\ell)\in\binom {[m]}\ell$, and some perturbation $\lambda_{i_1},\ldots,\lambda_{i_{\ell}}$, the rank of $A'$ increases. Then, from \Cref{bad} we know that $\mfI_{\vec a,I}$ is of size at most $n^{c\ell-c+1}$. Therefore, we enumerate over all $\lambda_{i_1},\ldots,\lambda_{i_\ell}\in S$ and check if $(\lambda_{i_1},\ldots,\lambda_{i_\ell})\in\mfI_{\vec a,I}$. This way, we find $j$ such that $(a_{i_1},\ldots,a_{i_\ell})$ is the $j^{th}$ tuple in $\mfI_{\vec a,I}$. Finally, we reset $\mcC\gets  (i_1\ldots,i_\ell,j,k,\vec a^{-I})$ where $\vec a^{-I}$ is $\{a_1\ldots,a_m\}\setminus\{a_{i_1},\ldots,a_{i_\ell}\}$. Notice that the index $j$ can be written with memory only $(c\ell-c+1)\log n$. Therefore, in this process, we have freed up $cm\log n-\ell\log n-(c\ell-c+1)\log n-\log n-c(m-\ell)\log n=(c\ell-\ell-c\ell+c-1-1)\log n=\mcO(\log n)$ since we have chosen $c=2\ell+3$.

		\item Suppose we just processed $\mcC_N$ and still did not find a required assignment. In this case, we have freed up at least $\mcO(n^3\log n)$ space for $N=n^3$. We shift the catalytic strings to make this space contiguous, and run \cite{BlaserJindalPandey2018} algorithm to find an assignment $A'$ achieves the $(1-\eps)$ approximation for $\rank(\mcA)$. Next set $index=N$ and go to step $4$.
		\item In either case, we now recompute the modified catalytic strings $\mcC_1,\ldots,\mcC_{index}$. For each $t\in[index]$, we execute \Cref{alga2} with $\mcC_t=(i_1\ldots,i_\ell,j,k,\vec a^{-I})$ as follows:
		
		We wish to compute the $j^{th}$ tuple from the set $\mfI_{\vec a,I}$. Notice that $$\mfI_{\vec a,I}=\Big\{\lambda_{i_1},\ldots,\lambda_{i_\ell}\in S:\rank \Big(\sum_{i\in[m]\setminus I}A_ia_i+\sum_{i\in I}A_i\lambda_i\Big)=k\Big\}.$$ Therefore we enumerate over all $\lambda_{i_1},\ldots,\lambda_{i_\ell}\in S$ until we find the $j^{th}$ tuple $(a_{i_1},\ldots,a_{i_\ell})\in\mfI_{\vec a,I}$. Finally we reset $\mcC\gets  (a_1,\ldots,a_m)$.
	\end{enumerate}
	The above algorithm clearly runs in polynomial time $\poly(n^\ell)$, and the catalytic space used is $\mcO(c\cdot mn^2\log n+|\mcB|)=\poly(n,\ell)$ with work space $\mcO(\ell \log n)$, and the catalytic tape is always restored when the algorithm halts. Thus the above gives the required $\CL\P$ algorithm that finds a matrix $A\in\mcA$ such that $\rank A\ge (1-\eps)\cdot \rank (\mcA)$.
\end{proof}
\begin{corollary}
	Given an instance of Linear Matroid Matching $M$, we have a $\CL\P$ algorithm that approximates the size of a maximum independent matching of $M$ upto a factor of $(1-\eps)$.
\end{corollary}
\begin{proof}
	Let $V$ be the ground set of the matroid and $T$ be the indeterminate Tutte matrix corresponding to the underlying graph on $V$ given by the instance $M$. Let $A$ be the linear representation matrix of $M$. Then \cite{murota2000matrices} shows that the size of the maximum independent matching in $M$ is exactly $1/2$ times $\rank (ATA^T)$. Therefore, \Cref{thm:edm} gives the desired result.
\end{proof}




\bibliographystyle{alpha}
\bibliography{biblio}

\end{document}